\documentclass[aps,pra,reprint,10pt,a4paper]{revtex4-1}
\usepackage[utf8]{inputenc}
\usepackage[sc,osf]{mathpazo}
\usepackage[T1]{fontenc}
\usepackage{amsfonts}
\usepackage{amssymb}
\usepackage{amsmath}
\usepackage{amsthm}
\usepackage{graphics}
\usepackage{graphicx}
\usepackage{braket}
\usepackage[colorlinks=true,linkcolor=blue,urlcolor=blue,citecolor=blue]{hyperref}

\newtheorem{proposition}{Proposition}
\newtheorem{remark}{Remark}

\newtheorem{observation}{Observation}

\begin{document}
\title{Characterizing the boundary of the set of absolutely separable states and their generation  via noisy environments}

\author{Saronath Halder$^1$, Shiladitya Mal$^{1,2}$, Aditi Sen(De)$^1$}

\affiliation{$^1$Harish-Chandra Research Institute, HBNI, Chhatnag Road, Jhunsi, Allahabad 211 019, India\\ $^2$Department of Physics and Center for Quantum Frontiers of Research and Technology (QFort), National Cheng Kung University, Tainan 701, Taiwan}

\begin{abstract}
We characterize the boundary of the convex compact set of absolutely separable states, referred as {\bf AS}, that cannot be transformed to entangled states by global unitary operators, in $2\otimes d$ Hilbert space.  However, we show that the absolutely separable states of rank-$(2d-1)$ are extreme points of such sets. We then discuss conditions to examine if a given full-rank absolutely separable state is an interior point or a boundary point of {\bf AS}. Moreover, we construct two-qubit absolutely separable states which are boundary points but not extreme points of {\bf AS} and prove the existence of full-rank extreme points of {\bf AS}. Properties of certain interior points are also explored. We further show that by examining the boundary of the above set,  it is possible to develop an algorithm to generate the absolutely separable states which stay outside the maximal ball. By considering paradigmatic noise models, we find  the amount of local noise which the input entangled states can sustain, so that the output states do not become absolutely separable. Interestingly, we report that with the decrease of entanglement of the pure input state, critical depolarizing noise value, transferring an entangled state to  an absolutely separable one, increases, thereby showing advantages of sharing nonmaximally entangled states. Furthermore, when the input two-qubit states are Haar uniformly generated,  we report a hierarchy among quantum channels according to the generation of absolutely separable states.
\end{abstract}
\maketitle

\section{Introduction}\label{sec1}
Characterization of resourceful quantum states is important from the perspective of several quantum information processing tasks \cite{Preskill98, Nielsen00}. These include quantum communication protocols like quantum state transfer \cite{Bennett93, Bouwmeester1997} using  entangled states \cite{Horodecki09-1, Guhne09}, encoding of classical information into quantum states \cite{Bennett92, Mattle96}, secure communication via entangled states \cite{Ekert91, Bostrom02, Beige02, Gisin02, Scarani09}, and measurement-based quantum computation \cite{Raussendorf01, Hein04, Briegel09}. In a resource theory, along with the characterization of set of states according to certain tasks, understanding the set of operations, known as free operations, by which resourceful states cannot be created is also important. For example, in the theory of entanglement, the set of local operations and classical communication constitute the free operations by which only separable states can be produced. Therefore, characterizing useless states in any paradigm can be essential to understand the free operations.

Two-qubit gates or joint unitary operators acting on two-qubit pure states can, in general, create entanglement in the systems \cite{Kraus01, Leifer03, Chefles05}. For example, from the initial  product state $|-\rangle|0\rangle$, with $|-\rangle = \frac{1}{\sqrt{2}}( |0\rangle - |1\rangle)$, a two-qubit {\it CNOT} gate can create a maximally entangled state. However, it was shown that there are bipartite states from which it is not possible to generate entanglement by acting joint (global) unitary operations. They are called {\it absolutely separable states} \cite{Kus01, Verstraete01} (see also Refs.~\cite{Hildebrand07, Slater09, Johnston13, Ganguly14, Arunachalam15, Jivulescu15} in this regard). In a resource theory in which global unitary operators are free operations, the set of absolutely separable states are not useful states. However, it is important to understand the properties of such states from a resource theoretic point of view. 

Recently, the witness operators  (for entanglement witnesses, see Refs.~\cite{Horodecki96, Lewenstein00, Terhal01-1}) have been constructed to separate the absolutely separable states from the separable ones (which are not absolutely separable) \cite{Ganguly14} by using the fact that the set, containing absolutely separable states, is convex and compact, so that the {\it Hahn-Banach separation theorem} \cite{Holmes75} can be applied.
In this context, we note that to construct optimal witness operators, it is important to explore the boundary points of the set of absolutely separable states. Moreover, the {\it Krein-Milman theorem} \cite{Krein40} states that a convex compact set corresponding to a finite dimensional vector space is equal to the convex hull of the extreme points of that set. Therefore, to understand a convex compact set, it is enough to know about the extreme points of that set. 

The main objectives of the present work is twofold:~(i) We consider the characterization of the boundary of the set of absolutely separable states when the quantum system is associated with a $2\otimes d$ Hilbert space. In particular, we show the states having rank-$(2d-1)$ in $2\otimes d$ always are extreme points of the set. On the other hand, the states having full-rank can be interior as well as extreme points of the above set, which can be on or outside of the maximal ball. By exploring the boundary of this set, it is also possible to develop an algorithm to generate the absolutely separable states which stay outside the {\it maximal ball} (it is the maximal ball around a maximally mixed state in which all the states are separable) \cite{Zyczkowski98, Gurvits02}. Moreover, we construct two-qubit absolutely separable states which are boundary points but not extreme points of the set of absolutely separable states. We then prove the existence of full-rank extreme points of that set (see Fig. \ref{diag} for schematic representation). (ii) We further search for noisy scenarios which result in absolutely separable states. Finding such situations can be interesting from the perspective of experiments. Specifically, we find critical strengths of local noise for different prototypical noise models \cite{Preskill98, Nielsen00}, which lead to absolutely separable states. Further, we find that when Haar uniformly  generated  two-qubit states are sent through noisy channels,  the process of generating  an absolutely separable state can distinguish three quantum channels, depolarizing, amplitude damping and phase damping channels. Since entangled states cannot be generated from absolutely separable states by applying global unitary operators, effects of local decoherence on state space induces irreversibility in the theory of entanglement.

\begin{figure}[h!]
\includegraphics[scale = 0.3]{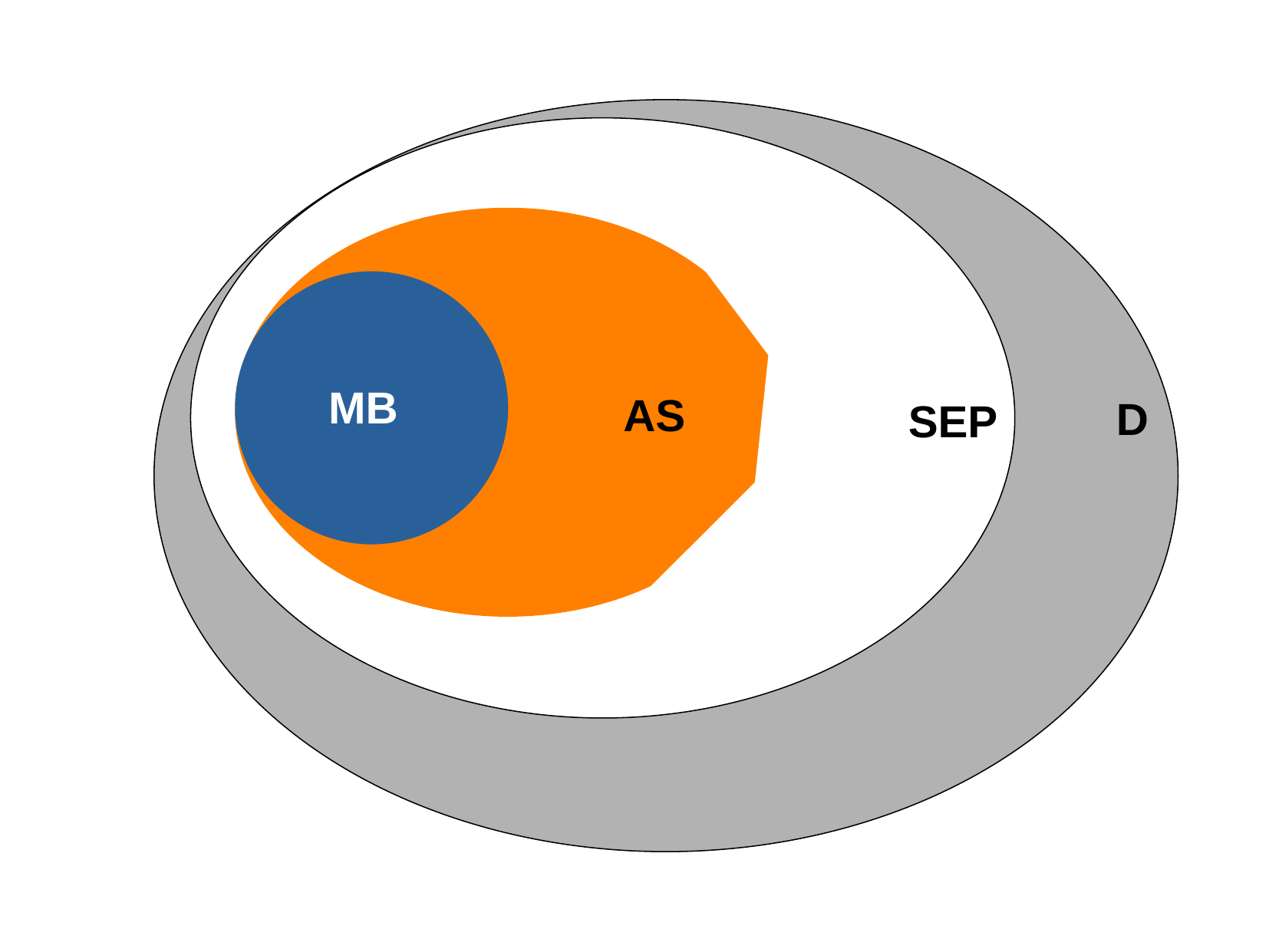}
\caption{(Color online) {\bf MB} represents the maximal ball. It is fully contained in {\bf AS}, the set of absolutely separable states, while the set of separable states is marked as {\bf SEP}. Similarly,  {\bf SEP} is the subset of {\bf D} denoting the entire state-space. In  $2\otimes2$, the portion where the boundaries of {\bf MB} and {\bf AS} are touching with each other, contains the rank-$3$ extreme points of {\bf AS} while  the region in the boundary of {\bf AS} which is not touching that of {\bf MB}, contains rank-$4$ extreme points of {\bf AS}. Furthermore, the line segments at the boundary of {\bf AS} represent the boundary points which are not extreme points.}
\label{diag}
\end{figure}

This paper is arranged as follows: After a few basics, presented in Sec.~\ref{sec2}, we characterize the set of absolutely separable states in Sec.~\ref{sec3}. In Sec.~\ref{sec4}, we address the question of irreversibility appearing due to the existence of absolutely separable states when an entangled state is affected by decoherence. Finally, we conclude in Sec.~\ref{sec5}.

\section{Preliminaries}\label{sec2}
Absolutely separable states are those which cannot be transformed to entangled states under the action of global unitary operators. For any given Hilbert space, the identity operator is an example of an absolutely separable operator. For a $2\otimes d$ Hilbert space, there is a {\it necessary and sufficient} condition to check whether a separable state is absolutely separable \cite{Verstraete01, Hildebrand07, Johnston13}. In particular, it was shown that in a \(2\otimes d\) system, a state is absolutely separable if and only if 
\begin{equation}\label{eq1}
\lambda_1 - \lambda_{2d -1} - 2 \sqrt{\lambda_{2d -2}\lambda_{2d}} \leq 0,
\end{equation}
where \(\lambda_1, \dots, \lambda_{2d} \) are the eigenvalues in decreasing order corresponding to a bipartite density matrix \(\rho_{AB}\), associated with a $2\otimes d$ Hilbert space.

With increasing importance of the entangled states, the properties of states lying in the neighborhood of the maximally mixed states \cite{Braunstein99} were studied. It was shown that the largest ball of separable as well as absolutely separable states around a maximally mixed state in a two-qubit system can be described by $\mbox{Tr}(\rho^2)\leq\frac{1}{3}$ \cite{Kus01, Zyczkowski98, Gurvits02}. This ball is known as a {\it maximal ball}. In general, for a $d\otimes d$ system, the maximal ball can be described by $\mbox{Tr}(\rho^2)\leq\frac{1}{d^2-1}$\cite{Zyczkowski98, Gurvits02}. However, it was found that there are absolutely separable states which reside outside the maximal ball \cite{Kus01}. So, to understand these states, it is important to construct such states.  

The possible structures of witness operators which separates absolutely separable states from the separable states were explored in \cite{Ganguly14}. But to make these operators optimal, one may require to explore the boundary points of the set of absolutely separable states. So, understanding these boundary points is one of the main objectives, discussed in the succeeding section. In this regard, remember that a state is not an extreme point of a convex set if it can be written as a convex combination of two or more absolutely separable states. On the other hand, if an absolutely separable state does not allow any such decomposition, then it must be an extreme point of the set.

\section{Characterization of Absolutely Separable States}\label{sec3}
Let us now concentrate on any two-party state, \(\rho_{AB}\) in \(2\otimes d\) which is absolutely separable. It is known that the pure product states can always be transformed to entangled states under global unitary operations. From condition (\ref{eq1}), it is also clear that rank-$2$ states cannot be absolutely separable states in $2\otimes d$ \cite{Verstraete01, Hildebrand07, Johnston13}. Moreover, the condition given in (\ref{eq1}), says that the mixed states having rank $\leq(2d-2)$ cannot be absolutely separable in $2\otimes d$. So, in $2\otimes2$, absolutely separable states can have rank $3$ and $4$. Since the set of absolutely separable states is a convex and compact set in any dimension \cite{Ganguly14}, it is possible to explore certain geometric properties of such a set, especially the structure of the extreme points of the set of absolutely separable states. Here, we call this set as {\bf AS}. We prove the following proposition for absolutely separable states having lowest rank in $2\otimes2$:

\begin{proposition}\label{prop1}
All rank-$3$ absolutely separable states are extreme points of {\bf AS} in $2\otimes 2$. 
\end{proposition} 

\begin{proof}
In $2\otimes2$, a state is absolutely separable if and only if $\lambda_1-\lambda_3-2\sqrt{\lambda_2\lambda_4}\leq0$, where $\lambda_1,\dots,\lambda_4$ are the eigenvalues of a given density matrix in decreasing order. For a rank-$3$ state, $\lambda_4=0$, which implies that a rank-$3$ state is absolutely separable iff $\lambda_1\leq\lambda_3$. But we have assumed that $\lambda_1\geq\lambda_3$. This indicates that a rank-$3$ state is absolutely separable if and only if $\lambda_1 = \lambda_2 = \lambda_3 = \frac{1}{3}$. Now, consider the spectral decomposition of a state, $\rho_{AB}$ = $\frac{1}{3}\ket{\psi_1}\bra{\psi_1} + \frac{1}{3}\ket{\psi_2}\bra{\psi_2} + \frac{1}{3}\ket{\psi_3}\bra{\psi_3}$, where $\{\ket{\psi_i}\}$, $i$=$1,2,3$ are orthogonal states. Obviously, this state is the only absolutely separable state in the three-dimensional subspace spanned by $\{\ket{\psi_1}, \ket{\psi_2}, \ket{\psi_3}\}$. Therefore, this state cannot be written as a convex combination of two or more absolutely separable states, proving the fact that all rank-$3$ absolutely separable states are extreme points of {\bf AS}. \end{proof} 

Since, there is only one particular structure, possible for rank-$3$ absolutely separable states, any rank-$3$ state which does not capture such a structure must not be an absolutely separable state. Now, we make the following remark: 

\begin{remark}\label{rem1}
The number of rank-$3$ extreme points of {\bf AS} in $2\otimes2$ can be infinitely many and they reside on the maximal ball.
\end{remark}

This is because in a two-qubit Hilbert space, one can choose three orthogonal states in infinitely many possible ways and then, take an equal mixture of three orthogonal states, leaving the final state to be an extreme point of {\bf AS}. Next, using the condition given in (\ref{eq1}), the generalization of the above proposition for a $2\otimes d$ system is presented: 
\begin{proposition}\label{coro1} All rank-$(2d-1)$ absolutely separable states are extreme points of {\bf AS} in $2\otimes d$. \end{proposition}
\begin{proof}
The proof is along the same line as Proposition \ref{prop1}. It uses the fact that any state of {\bf AS}, having rank-$(2d-1)$, is the only state of {\bf AS} supported in a proper subspace (of $2\otimes d$ Hilbert space) spanned by a set of orthogonal states $\{\ket{\psi_i}\}_{i=1}^{2d-1}$. 
\end{proof}

The above state has a spectral decomposition: $\rho$ = $\frac{1}{2d-1}(\sum_{i=1}^{2d-1}\ket{\psi_i}\bra{\psi_i})$, where $\ket{\psi_i}$s are orthogonal states. Clearly, purity of the state $\rho$ is given by $\mbox{Tr}[\rho^2]$ = $\frac{1}{2d-1}$. In the context of Proposition \ref{prop1} and Proposition \ref{coro1}, see also Refs.~\cite{Arunachalam15,  Jivu15}.

\begin{remark}\label{rem2}
All rank-$(2d-1)$ absolutely separable states in $2\otimes d$ having purity $\frac{1}{2d-1}$ are included in the maximal ball characterized by $[\mbox{Tr}(\rho_{AB}^2)\leq\frac{1}{2d-1}]$ and they all lie on the surface of the ball.
\end{remark}

Interestingly, there exist rank-$4$ absolutely separable states in $2\otimes2$ which can stay outside the maximal ball \cite{Kus01}. In this context, we mention that there are ways to check if a given state is an absolutely separable state and if it belongs to the maximal ball but there is no known protocol to produce rank-$4$ absolutely separable states systematically in $2\otimes2$ which reside outside the maximal ball. We now propose a prescription to produce rank-$4$ absolutely separable states which reside outside the maximal ball. This can be done by exploring the boundary of {\bf AS}. 

\begin{enumerate}
\item Take a rank-3 state in $2\otimes 2$ which is not absolutely separable, i.e., such states must reside outside the maximal ball. From Proposition \ref{prop1}, it is clear that such states are of the form $\sum_{i=1}^3p_i\ket{\psi_i}\bra{\psi_i}$ with at least one $p_i$, not equal to the other $p_i$s.

\item Consider a pure state in $2\otimes2$ which is orthogonal to the previous state. Pure states are not included in the maximal ball and they cannot be absolutely separable either.

\item A suitable convex combination of these two states can produce rank-$4$ states in $2\otimes2$ which are absolutely separable. But we have to take the convex combination in such a way that the newly generated states reside outside the maximal ball for some choices of parameters.
\end{enumerate}

\paragraph*{Example.}Let us now illustrate the recipe discussed above by an example. We consider a two-qubit rank-$3$ state $\rho_1$ = $\frac{1}{2}\ket{00}\bra{00}+\frac{1}{4}\ket{01}\bra{01}+\frac{1}{4}\ket{10}\bra{10}$, which is not an absolutely separable state. Also consider a pure state $\rho_2$ = $\ket{11}\bra{11}$ which cannot be absolutely separable. These two states are orthogonal to each other. So, any convex combination of them must be a rank-$4$ state. We now take convex combination of these two states $q\rho_1+(1-q)\rho_2$ in a way that $q=\frac{16}{17}$. It can be checked that the newly prepared state is an absolutely separable state. Interestingly, if $q=\frac{16}{17}$, then $\mbox{Tr}[q\rho_1+(1-q)\rho_2]^2>\frac{1}{3}$, indicating that the state resides outside the maximal ball. We observe that the value of \(q\) is not independent of the choice \(\rho_1\).

Notice that in the above, when $q=\frac{16}{17}$, the newly prepared rank-$4$ state is just included into {\bf AS}. Therefore, it is a boundary point of the set. But it is not known whether the state is an extreme point of {\bf AS}. Note that the above protocol is quite easy to generalize for $2\otimes d$. In that case, one has to start with a rank-$(2d-1)$ state which is not included in the maximal ball as well as in {\bf AS} (this can be found by Proposition \ref{coro1}). Then consider a pure state which is orthogonal to the previous state and the rest is as described above. 

Let us now move to full-rank, i.e., rank-$4$ absolutely separable states in $2\otimes2$ and discuss different properties of a set consisting of such states. It is known that there exists a ball around the maximally mixed state, ($\mathbb{I}/d^2$), in $d\otimes d$ and all the states, belonging to that ball, are absolutely separable \cite{Zyczkowski98, Gurvits02}. So, the maximally mixed state is an interior point of {\bf AS} in $d\otimes d$. Clearly, any state $\rho$ (except the maximally mixed state) which is an interior point of {\bf AS}, can be written as a convex combination of the maximally mixed state ($\mathbb{I}/d^2$) and another absolutely separable state $\sigma$, where $\rho\neq\sigma$, and both $\rho$ and $\sigma$ are not maximally mixed. Note that for different states $\rho$, the states $\sigma$ can be different. However, the states which do not allow such a decomposition must be boundary points of {\bf AS}. We now present the following observation: 

\begin{proposition}\label{obs1}
In $2\otimes2$, any rank-$4$ absolutely separable state $\rho$ which satisfy the condition $\lambda_1-\lambda_3$ = $2\sqrt{\lambda_2\lambda_4}$, cannot be written as a convex combination of the maximally mixed state $(\mathbb{I}/4)$ and another absolutely separable state $\sigma$, not maximally mixed, with $\rho\neq\sigma$, where $\lambda_i$s are the eigenvalues of $\rho$ in decreasing order. \end{proposition}
 
\begin{proof}
We consider a spectral decomposition of the above state $\rho$ = $\sum_{i=1}^4\lambda_i\ket{\psi_i}\bra{\psi_i}$. Now, consider a decomposition of $\rho$ as the following:
\small
\begin{equation}\label{eq2}
\begin{array}{l}
\rho = \sum_{i=1}^4\lambda_i\ket{\psi_i}\bra{\psi_i} = (1-4\epsilon)\sum_{i=1}^4\frac{(\lambda_i-\epsilon)}{1-4\epsilon}\ket{\psi_i}\bra{\psi_i}\\[1 ex]
+ (4\epsilon)\sum_{i=1}^4\frac{\epsilon}{4\epsilon}\ket{\psi_i}\bra{\psi_i} = (1-4\epsilon)\sigma + (4\epsilon)\frac{\mathbb{I}}{4},
\end{array}
\end{equation} 
\normalsize
where $\epsilon$ can be considered as a very small number and $\mathbb{I}$ is the identity operator. Now, the question is whether $\sigma$ is an absolutely separable state. The state $\sigma$ is absolutely separable if and only if the eigenvalues of it $(\lambda_i-\epsilon)/(1-4\epsilon)$ (in decreasing order) obey the condition (\ref{eq1}). If it is the case then $\lambda_1-\lambda_3\leq2\sqrt{(\lambda_2-\epsilon)(\lambda_4-\epsilon)}$. But this cannot be because we have assumed $\lambda_1-\lambda_3$ = $2\sqrt{\lambda_2\lambda_4}$. So, for any $\epsilon>0$, $(\lambda_1-\lambda_3)>2\sqrt{(\lambda_2-\epsilon)(\lambda_4-\epsilon)}$, leading to the fact that $\sigma$ cannot be absolutely separable. Thus, it is proved that if the eigenvalues of $\rho$ (in the decreasing order) obey the condition $\lambda_1-\lambda_3$ = $2\sqrt{\lambda_2\lambda_4}$ then the state cannot be written as a convex combination of the maximally mixed state and another absolutely separable state. \end{proof}

In the above, it is clear that any state $\rho$ whose eigenvalues (in the decreasing order) obey the condition $\lambda_1-\lambda_3$ = $2\sqrt{\lambda_2\lambda_4}$, are the boundary points of {\bf AS}. Next, we describe an important property of certain interior points of the set.

\begin{remark}
\label{rem_3}
It is interesting to note that by using similar arguments,  Proposition \ref{obs1} can be extended to $2\otimes d$. Let us consider the spectral decomposition of a state $\rho$ = $\sum_{i=1}^{2d}\lambda_i\ket{\psi_i}\bra{\psi_i}$, where  $\lambda_i$ are in decreasing order. Now, consider a decomposition of $\rho$ as 
\begin{equation}\label{eq2}
\begin{array}{l}
\rho = \sum_{i=1}^{2d}\lambda_i\ket{\psi_i}\bra{\psi_i} = (1-2d\epsilon)\sum_{i=1}^{2d}\frac{(\lambda_i-\epsilon)}{1-2d\epsilon}\ket{\psi_i}\bra{\psi_i}\\[1 ex]
+ (2d\epsilon)\sum_{i=1}^{2d}\frac{\epsilon}{2d\epsilon}\ket{\psi_i}\bra{\psi_i} = (1-2d\epsilon)\sigma + (2d\epsilon)\frac{\mathbb{I}}{2d},
\end{array}
\end{equation} 
where $\epsilon$ can be considered as a very small number and $\mathbb{I}$ is the identity operator. Now, the question is whether $\sigma$ is an absolutely separable state. The state $\sigma$ is absolutely separable if and only if the eigenvalues of it $(\lambda_i-\epsilon)/(1-2d\epsilon)$ (in the decreasing order) obey the condition, given by 
\begin{equation}
\begin{array}{c}
\frac{(\lambda_1-\epsilon)}{(1-2d\epsilon)}-\frac{(\lambda_{2d-1}-\epsilon)}{(1-2d\epsilon)}\leq2\sqrt{\frac{(\lambda_{2d-2}-\epsilon)}{(1-2d\epsilon)}\frac{(\lambda_{2d}-\epsilon)}{(1-2d\epsilon)}}.
\end{array}
\end{equation}
If it is the case then $\lambda_1-\lambda_{2d-1}\leq2\sqrt{(\lambda_{2d-2}-\epsilon)(\lambda_{2d}-\epsilon)}$. But this cannot be true if we begin with $\lambda_1-\lambda_{2d-1}$ = $2\sqrt{\lambda_{2d-2}\lambda_{2d}}$. The reason behind it is that for any $\epsilon>0$, $\lambda_1-\lambda_{2d-1}>2\sqrt{(\lambda_{2d-2}-\epsilon)(\lambda_{2d}-\epsilon)}$, leading to the fact that $\sigma$ cannot be absolutely separable. Thus, it is proved that if the eigenvalues of $\rho$ (in the decreasing order) obey the condition $\lambda_1-\lambda_{2d-1}$ = $2\sqrt{\lambda_{2d-2}\lambda_{2d}}$,  the state cannot be written as a convex combination of the maximally mixed state and another absolutely separable state. This also implies that the states, which satisfy the above equality, are boundary points of {\bf AS} when the associating Hilbert space is $2\otimes d$.
\end{remark}

\begin{proposition}\label{prop2}
In $2\otimes2$, if the eigenvalues (in the decreasing order) of a rank-$4$ absolutely separable state (except the maximally mixed state) satisfy the condition $\lambda_1 - \lambda_{3} < 2 \sqrt{\lambda_{2}\lambda_{4}}$ the state can be written as a convex combination of the maximally mixed state and a boundary point of {\bf AS}. 
\end{proposition}

\begin{proof}
We consider the states other than the maximally mixed state. Suppose that the spectral decomposition of such a state is given by $\rho$ = $\sum_{i=1}^4\lambda_i\ket{\psi_i}\bra{\psi_i}$. This state can be decomposed into the following form:
\small
\begin{equation}\label{eq3}
\begin{array}{l}
\rho = \sum_{i=1}^4\lambda_i\ket{\psi_i}\bra{\psi_i} = (1-4\epsilon)\sum_{i=1}^4\frac{(\lambda_i-\epsilon)}{1-4\epsilon}\ket{\psi_i}\bra{\psi_i}\\[1 ex]
+ (4\epsilon)\sum_{i=1}^4\frac{\epsilon}{4\epsilon}\ket{\psi_i}\bra{\psi_i} = (1-4\epsilon)\sigma + (4\epsilon)\frac{\mathbb{I}}{4},
\end{array}
\end{equation} 
\normalsize
Our goal is to prove that the state $\sigma$ in the above decomposition is a boundary point of {\bf AS}. Because if it is the case then the absolutely separable state $\rho$ can be written as the convex combination of the maximally mixed state and a boundary point of {\bf AS}. Notice that the eigenvalues of $\sigma$ are given by $(\lambda_i-\epsilon)/(1-4\epsilon)$. It is a boundary point of {\bf AS} if $\lambda_1-\lambda_3=2\sqrt{(\lambda_2-\epsilon)(\lambda_4-\epsilon)}$. We have assumed that $\lambda_1-\lambda_3<2\sqrt{\lambda_2\lambda_4}$. Clearly, it is possible to consider $\epsilon\leq\lambda_4$, such that $\sigma$ becomes a boundary point. The value of $\epsilon$ can be found by solving the quadratic equation $\lambda_1-\lambda_3 = 2\sqrt{(\lambda_2 - \epsilon)(\lambda_4 - \epsilon)}$. This completes the proof.
\end{proof}

Obviously, the states of Proposition \ref{prop2} are interior points of {\bf AS}. By Propositions \ref{obs1} and \ref{prop2}, we  analyzed the known necessary and sufficient condition, given in (\ref{eq1}), to a further extent for two qubits. This condition was to check whether a given state is absolutely separable or not. In our case, we have established the conditions to find out  whether a given absolutely separable state is a boundary point (when the equality holds) or an interior point (when the inequality holds) of the set consisting of  absolutely separable states. However, like Proposition \ref{prop1}, Proposition \ref{prop2} can also be generalized in $2\otimes d$ by following the same argument given in the above proof and hence  we have the following:

\begin{proposition}\label{coro2} For a $2\otimes d$ system if the eigenvalues (in decreasing order) of a given absolutely separable state (except the maximally mixed state) satisfy the condition $\lambda_1-\lambda_{2d-1}<2\sqrt{\lambda_{2d-2}\lambda_{2d}}$, the state can be written as a convex combination of the maximally mixed state and a boundary point of {\bf AS}. Obviously, such a state must be an interior point of {\bf AS}. \end{proposition}

Let us now address another important question: Are all boundary points of {\bf AS} extreme points? In the following, we show that this is not the case.

\begin{proposition}\label{prop3}
In $2\otimes2$, there exist rank-$4$ absolutely separable states which are the boundary points of {\bf AS} but they are not extreme points of the set. 
\end{proposition}

\begin{proof}
We simply construct a class of states, the eigenvalues (in the decreasing order) of which satisfy the condition (\ref{eq1}) with equality and the constructed states allow some convex decomposition. Specifically, let us consider a two-qubit state $\sigma_1 = \sum_{i=1}^4\lambda_i\ket{\psi_i}\bra{\psi_i}$ (spectral decomposition), where $\lambda_i$s are in decreasing order and they satisfy the condition $\lambda_1-\lambda_3$ = $2\sqrt{\lambda_2\lambda_4}$. Similarly, we have another two-qubit state $\sigma_2 = \sum_{i=1}^4\lambda^\prime_i\ket{\psi_i}\bra{\psi_i}$ (spectral decomposition), with $\lambda^\prime_i$s being in decreasing order, satisfying the condition $\lambda^\prime_1-\lambda^\prime_3$ = $2\sqrt{\lambda^\prime_2\lambda^\prime_4}$. For any convex combination $x\sigma_1+(1-x)\sigma_2$, $0<x<1$, the newly generated state can have the following eigenvalues: $\mu_i$ = $x\lambda_i + (1-x)\lambda^\prime_i$, for $i=1,\dots,4$. $\mu_i$s are also in decreasing order. It can be shown that $\mu_1-\mu_3$ = $2\sqrt{\mu_2\mu_4}$ if $\lambda^\prime_2/\lambda^\prime_4$ = $\lambda_2/\lambda_4$. We assume $\lambda^\prime_2/\lambda^\prime_4$ = $\lambda_2/\lambda_4$ = $\kappa$. Using this along with the conditions $\sum_i\lambda_i=\sum_i\lambda^\prime_i=1$, it can be shown that $\lambda_1$ = $[1-(1+\kappa)\lambda_4+2\sqrt{\kappa}\lambda_4]/2$, $\lambda_2$ = $\kappa\lambda_4$, $\lambda_3$ = $[1-(1+\kappa)\lambda_4-2\sqrt{\kappa}\lambda_4]/2$. $\lambda^\prime_i$s also satisfy similar relations. [Notice that for $i=1,2,3$, $\lambda_i$s or $\lambda_i^\prime$s are the function of $\lambda_4$ or $\lambda_4^\prime$ respectively. Similarly, for $i=1,2,3$, $\mu_i$s are also the same function of $\mu_4$.] However, for proper choice of $\kappa$, $\lambda_4$, $\lambda^\prime_4$, one can get $\lambda_i$s and $\lambda^\prime_i$s in decreasing order. Thus, one can generate absolutely separable states which satisfy the condition given in (\ref{eq1}) with equality. Moreover, these states allow convex decomposition, implying the fact that such states are boundary points but not extreme points of {\bf AS}. To constitute an example, one may consider $\kappa = 2.5$, $\lambda_4 = 0.1$, and $\lambda^\prime_4 = 0.11$. 
\end{proof}

In the above context, we mention that if we assume for $\sigma_1$, $\lambda_1-\lambda_3 < 2 \sqrt{\lambda_2 \lambda_4}$ and for $\sigma_2$, $\lambda^\prime_1-\lambda^\prime_3<2\sqrt{\lambda^\prime_2\lambda^\prime_4}$, for any newly generated state, $\mu_1-\mu_3$ must be less than $2\sqrt{\mu_2\mu_4}$. Nevertheless, which states are the rank-$4$ extreme points of {\bf AS} in $2\otimes2$ is still an open problem but we are able to prove the existence of such states in the succeeding proposition.

\begin{remark}
Applying the  Remark \ref{rem_3} along with Proposition \ref{prop3}, we can now construct examples of boundary points which are not extreme points of {\bf AS} in $2\otimes d$. An explicit example can be given as follows:
We take $\nu_{2d}$ = $0.1$, $\nu_1$ = $[1-(1+\kappa)\nu_{2d}+2\sqrt{\kappa}\nu_{2d}]/2$, $\nu_{2d-2}$ = $\kappa\lambda_{2d}$, $\nu_{2d-1}$ = $[1-(1+\kappa)\nu_{2d}-2\sqrt{\kappa}\nu_{2d}]/2$. We can assume $\kappa = 2.5$. We next take $\nu_2=\cdots=\nu_{2d-3}=\kappa^\prime$, $\nu_{1}\geq\kappa^\prime\geq\nu_{2d-2}$ and then normalize the spectrum $\{\nu_i\}$ such that the normalized quantities $\{\nu_i^\prime\}$ satisfy the condition $\sum_{i=1}^{2d}\nu_i^\prime=1$. So, any density matrix, having eigenvalues $\{\nu_i^\prime\}$, is a boundary point of ${\bf AS}$.
In the similar fashion, if we take $\nu_{2d}$ = $0.11$ and follow the above procedure,  we can obtain another boundary point of ${\bf AS}$. Now, taking a suitable convex combination of these two boundary points, we can get a third boundary point which is surely not an extreme point of ${\bf AS}$ (as shown in  Fig. \ref{diag}).
\end{remark}

\begin{proposition}\label{prop4}
In $2\otimes2$, there exist rank-$4$ absolutely separable states outside the maximal ball which are extreme points of {\bf AS}. 
\end{proposition}

\begin{proof}
In Proposition \ref{prop1}, we have proved that all rank-$3$ absolutely separable states are extreme points of {\bf AS}. Moreover, they reside on the surface of {\it maximal ball}, defined by $\mbox{Tr}[\rho^2]\leq\frac{1}{3}$. We mention here that the maximal ball is a convex set, based on the fact that a convex combination does not allow to increase purity. We know that there are rank-$4$ absolutely separable states outside the maximal ball as depicted in Fig. \ref{diag}. An example of such a state is explicitly constructed after Proposition \ref{prop1}. Obviously, these rank-$4$ states cannot be written as a convex combination of rank-$3$ extreme points of {\bf AS} since these states reside outside the maximal ball. Therefore, either such a state is an extreme point of {\bf AS}, or they can be written as a convex combination of rank-$4$ extreme points of {\bf AS} out side the maximal ball. This completes the proof.
\end{proof}

\section{Interconvertability between absolutely separable and entangle state}\label{sec4}
Entanglement is a resource for different quantum information processing tasks. Since absolutely separable states cannot be transformed to entangled states under global unitary operations, their existence puts a restriction on the state space. First, we discuss how auxiliary systems can help to overcome absolute separability and prescribe a method to identify operations on {\bf AS} so that it becomes entangled. Second, we address how an entangled state get converted into separable or absolutely separable state through noisy channels.

\subsection{Qubit assisted entanglement generation}\label{sec4sub1}
It is easy to notice that given an absolutely separable state if one enlarges the local dimension by considering auxiliary qubit(s), one may find unitary operator(s), acting on the newly generated higher dimensional state, which can produce entanglement. So, the newly generated state, is basically the given absolutely separable state along with the auxiliary qubit(s), which corresponds to the higher dimensional Hilbert space. Now the question is: which kind of auxiliary qubit(s), one should take along with a given absolutely separable state such that at least one unitary operator exists, which acts on the newly generated state and can produce entanglement. Clearly, there might be a restriction on the form of the auxiliary qubit(s), depending on the given absolutely separable state, which can create entangled states. In this regard, we present the following: 

\begin{observation}\label{obs2} 
If the absolutely separable state is of full-rank, there can be a restriction on the form of the auxiliary qubit, while for a non-full-rank absolutely separable state, any auxiliary qubit  together with an absolutely separable state can   produce entangled states
by applying global unitary operator(s) which are being applied on the newly generated state in the higher dimensional Hilbert space.
\end{observation} 

The first part is simple to establish. Suppose, for a two-qubit system, the maximally mixed state is given. Obviously, that state is an absolutely separable state. Now, if the second party enlarges the local dimension by considering an auxiliary qubit, then it must not be prepared in the maximally mixed state in order to generate entanglement by applying a global unitary operator. Because if the auxiliary qubit is prepared in a maximally mixed state, the overall state is again a maximally mixed state in $2\otimes4$ and thus, generating entanglement using any global unitary operator is not possible. 

On the other hand, suppose, a non-full rank absolutely separable state is given in $2\otimes2$. Then any auxiliary qubit can be chosen if one wants to enlarge the local dimension, so that a separable state of rank lower than $2d-1$ can be obtained, which by construction is not absolutely separable. So, using the final state, entanglement can be generated via a suitable global unitary operator. This follows from the fact that in $2\otimes d$, a state which have rank less than $2d-1$ cannot be absolutely separable. 

From the above observation,  we prescribe a possible way of creating entangled state from {\bf AS} and derive Kraus operators for that kind of operations. Suppose a global unitary operator, \(U\),  acts both on an absolutely separable state and  an  auxiliary state $|0\rangle\langle 0|$. Mathematically, we can say the following:
\begin{equation}
\small
\begin{array}{c}
\Lambda(\rho_{AS})=Tr_B[U(\rho_{AS}\otimes |0\rangle_B\langle 0|)U^{\dagger}]=\sum_{\mu} K_{\mu}\rho_{AS}K_{\mu}^{\dagger},
\end{array}
\end{equation}
where $K_{\mu}=\langle \mu|U|0\rangle$ with $\sum K_{\mu}^{\dagger} K_{\mu}=\mathbb{I}$.

\subsection{Identifying  absolutely separable states in noisy environments}
\label{sec4sub2}
In realistic scenarios, due to imperfections in preparation procedure or interaction with the environment, the entangled state generated in the laboratory is not a pure one. This effect of noise can be modeled in various ways. Typically, a local noisy channel destroys entanglement, and thereby creates a separable state \cite{Horodecki09-1} which may or may not be absolutely separable. We are interested in finding the amount of noise required in the channel to produce absolutely separable states. Such a study can be important for two reasons -- (1) in experiments, a state is always affected by noise and (2) determining the range of noise parameter which makes the state  absolutely separable, is significant to avoid them.

We deal with three paradigmatic noise models: depolarizing, amplitude damping, and phase damping channels which can affect an initial state differently and, thereby, generate absolutely separable states in independent ways. Before discussing the consequence of noisy channels on a given state, let us first fix the transformation of the input state that happens due to the interactions between the environment and the system.  
The depolarizing channel (DPC), \(  \Lambda_{DPC}\) takes an arbitrary quantum state, \(\rho\) to \(\rho' = \Lambda_{DPC} (\rho) =  p \rho + \frac{(1-p)}{3} \sum_{i=x, y, z} \sigma_i \rho \sigma_i\) where \(\sigma_i\) \((i=x, y, z) \)s  represent the Pauli operators and (\(1-p\)) is the strength of the noise.  

On the other hand, the amplitude damping channel (ADC), $\Lambda_{ADC}$ acts asymmetrically on the states \(|0\rangle \) and \(|1\rangle\). In particular, it keeps \(|0\rangle\) unchanged while  it flips  \(|1\rangle\) to  \(|0\rangle\) with probability \((1-p)\) and the corresponding Kraus operators which describe the effects of ADC are given by
\begin{equation}\label{eq4}
K_{1} = \left(\begin{array}{cc}
1 & 0\\
0 & \sqrt{p}
\end{array}\right),~~
K_{2}=\left(\begin{array}{cc}
0 & \sqrt{1-p}\\
0 & 0
\end{array}\right).
\end{equation}
The input state \(\rho\) after sending through the ADC results in an output state, represented as
\begin{equation}\label{eq5}
\Lambda_{ADC}(\rho) =\sum_{i=1}^{2} K_{i}\rho K^{\dagger }_{i}. 
\end{equation}
Finally, we consider a scenario in which a qubit is sent through a phase damping channel (PDC), described by
\begin{equation}\label{eq6}
\Lambda_{PDC}( \rho) =  \sum_{i=1}^3 K_{i}\rho K^{\dagger }_{i}, 
\end{equation}
with the operators $K_i$, $i=1,2,3$ given by
\begin{equation}\label{eq7}
\begin{array}{c}
K_1 = \sqrt{1-p}\left(\begin{array}{cc}
1 & 0\\
0 & 1
\end{array}\right),~~
{K_2} = \sqrt{p}\left(\begin{array}{cc}
1&0\\
0&0
\end{array} \right),\\[3ex]
{K_3} = \sqrt {p}\left(\begin{array}{cc}
0&0\\
0&1
\end{array}\right).
\end{array} 
\end{equation}
A noisy environment always takes an initial pure state, \(|\psi\rangle = \cos \frac{x}{2} |00\rangle + e^{-i \phi} \sin \frac{x}{2}|11\rangle\) with \(0 \leq x \leq \pi\) and \(0 \leq \phi \leq 2 \pi \) to a mixed one which may or may not be an entangled state. Note that the input  state \(|\psi \rangle \)  is entangled for all values of \(x\) and \(\phi\) except when \(x  = 0 \) or \(\pi\).  We assume here that two independent and identical channels, \(\Lambda^1 (p) \otimes \Lambda^2 (p)\), act on the input state where \(\Lambda^i\) can be either DPC, ADC, or PDC, and we are interested with the properties of output state. 

\subsubsection{Absolutely separable states via depolarizing channels}

\begin{figure}
\includegraphics[scale = 0.6]{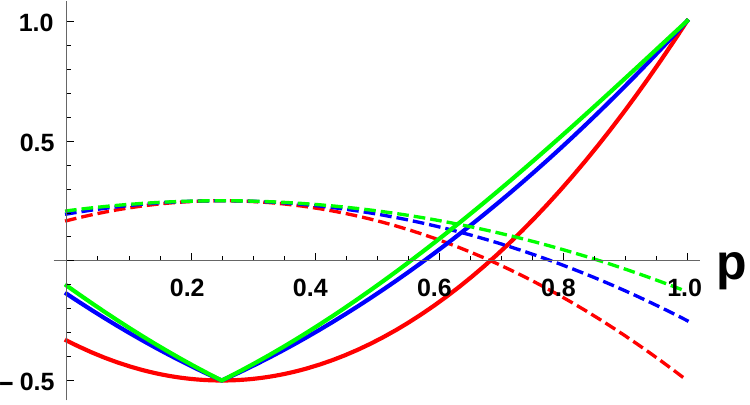}
\caption{(Color online) Creation of absolutely separable states from a pure state \(|\psi \rangle = \cos (x/2)|00\rangle + e^{-i \phi} \sin (x/2)|11\rangle\). Note that $|\psi\rangle$ is maximally entangled with $x=\pi/2$. It is sent through local depolarizing channels which produce separable as well as absolutely separable states with the variation of parameter \(p\) (abscissa). Red, blue, and green lines correspond to $x=\pi/2$, $x=\pi/6$, and $x=\pi/12$ respectively. Dashed lines correspond to  the minimum eigenvalues of the partial transposed output states for different values of \(x\) (ordinate) , thereby quantifying the entanglement contents of the output states while  solid lines represent the  quantity $\lambda_1-\lambda_3-2\sqrt{\lambda_2\lambda_4}$ (ordinate) for examining the absolute separability in $2\otimes2$, given in (\ref{eq1}). It is clear that there exists a range of $p$ where the output state is separable but not absolutely separable when the input state is a pure nonmaximally entangled state while for maximally entangled state, the critical value of noise above which state is separable as well as absolutely separable coincide (see red solid and dashed lines). Both axes are dimensionless.} \label{DPC}
\end{figure}

Let us start the investigation by studying the features of the output state obtained after sending the initial state via local depolarizing channels (also see Ref.~\cite{Filippov17} in this regard):
\begin{equation}\label{eq8}
|\psi \rangle \rightarrow \Lambda^1_{DPC} \otimes  \Lambda^2_{DPC} (|\psi\rangle \langle \psi|) = \rho'(p, x, \phi).\end{equation} 
The characteristics of the resulting state can be summarized as the following: 
\begin{enumerate}
\item The output state  is a rank-$4$ state. As discussed in Sec.~\ref{sec3}, the full-rank states can in principle be an absolutely separable for certain choices of \(p\), \(x\), and \(\phi\). For any fixed values of $x$, we find that there exists a critical value of noise above which the output state is absolutely separable.

\item Two eigenvalues of the output state are the same and are independent of the input state parameter, \(x\) and \(\phi\), while the other two  depend on \(x\) and \(p\). For a fixed value of \(x\), we observe that there exists a critical value of noise \(1-p_{abs}\), above which the state is absolutely separable (see Fig.~\ref{DPC}).

\item The partial transposed output state leads to the condition that the state is entangled when \(\sin x > [4(1+2p)(1-p)]/(4p-1)^2\) as depicted in Fig.~\ref{DPC}. 
\end{enumerate}
Remember that we have taken here the strength of the noise as $(1-p)$ while in this figure, the plots are made with respect to the parameter $p$. However, looking at properties 2 and 3, we observe that when \(x < \pi/2\), \((1-p_{abs})\) is strictly greater than \((1-p_{sep})\), below which the state is entangled. Interestingly, when \(x =\pi/2\), i.e., the input state is maximally entangled, $(1-p_{abs}) = (1-p_{sep})$. The gap obtained between the threshold values of noise for separability and absolute separability when the input state is a nonmaximally entangled state, can be interesting. We know that nonmaximally entangled states are less  useful than maximally entangled states in several quantum information protocols. Nonetheless, such disadvantages can be compensated since in the presence of a certain amount of noise, separable states that are not absolutely separable are created which can be converted to entangled states under global unitary operations. This is shown in Table \ref{tab1}. 

We also observe that if the entanglement of the input state decreases, the output state becomes separable quickly under the action of local DPCs. But lowering the input state entanglement, the range of the noise in which the output state remains only separable (not absolutely separable) increases.
 Specifically, for maximally entangled and nonmaximally entangled states, denoted by  \(|\psi_{\max}\rangle\) and  \(|\psi_{\mbox{nm}}\rangle\) respectively,  we obtain that \[(1-p_{sep}) (|\psi_{\max}\rangle) > (1 - p_{sep})  (|\psi_{\mbox{nm}}\rangle)\] while \[(1-p_{abs}) (|\psi_{\max}\rangle) < (1 - p_{abs})  (|\psi_{\mbox{nm}}\rangle).\] 
\begin{table}[ht]
\centering
\begin{tabular}{ |c|c|c|c|} 
\hline
& & & \\
\textbf{\textit{Entanglement of input}} & $1-p_{sep}$  & $1-p_{abs}$ & $\Delta p = p_{sep} - p_{abs}$ \\
& & &  \\
\hline
$0.7715$ & $0.29133$ & $0.36114$ & $0.0698$ \\ 
\hline
$0.33225$ & $0.21413$ & $0.426103$ & $0.21197$ \\ 
& & & \\
\hline
\end{tabular}
\caption{Entanglement of the input state is measured in von Neumann entropy \cite{Bennett00}. We show that with the decrease of entanglement of the input state, the range of  separability but not absolute separability increases.}\label{tab1}
\end{table}

\begin{remark}\label{rem3}
A maximally entangled state sent through a global depolarizing channel results in the Werner state \cite{Werner89}, given by \(p |\psi^-\rangle \langle \psi^-| + \frac{(1 -p)}{4} I\) (with \(|\psi^-\rangle\) being the singlet state) which is separable as well as absolutely separable with  \(p\leq 1/3\). On the other hand, if a nonmaximally entangled state is admixed with the white noise, the gap in the strength of noise between separability and absolute separability emerges like the local depolarizing channels, establishing the usefulness of nonmaximally entangled states over the maximally entangled ones.
\end{remark}

\emph{Generation of absolutely separable states from Haar uniformly simulated inputs.} Let us  generate two-qubit pure states Haar uniformly, and  both qubits  are sent through local depolarizing channels. We find a critical value of noise below which the states become absolutely separable. In particular, we plot  $p_c$ below which the state is absolutely separable by varying  the initial entanglement content of the pure states. Results show that with respect to noise, states with higher entanglement are more robust than the states with low entanglement value from the perspective of generation of  ASs (see Fig.~\ref{fig:HaarAS}).

\begin{figure}[h!]
\includegraphics[height=0.28\textheight, width=0.5\textwidth]{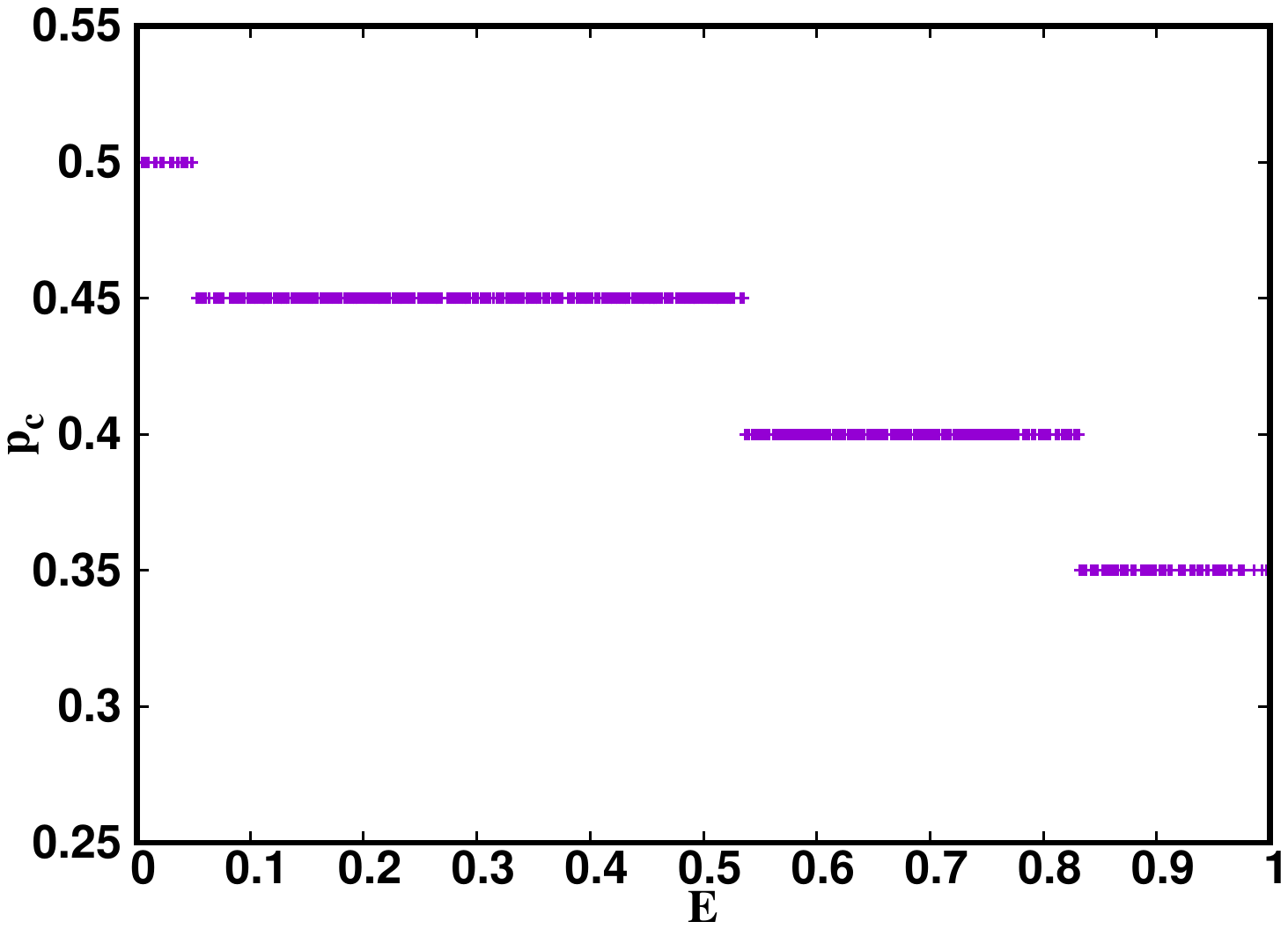}
\caption{(Color online.) Plot of $ p_c$ (critical noise value below which the state is absolutely separable) (vertical axis) against initial entanglement, $E$ (horizontal axis). We generate  \(10^4\)  pure states Haar uniformly.  Both qubits  are sent through the  local depolarizing channels. Entanglement of a pure state is characterized by the von Neumann entropy of the local density matrices. The vertical axis is dimensionless while the horizontal one is in ebits.}
\label{fig:HaarAS}
\end{figure}

\subsubsection{Absolutely separable states via amplitude damping channels}

\begin{figure}
\includegraphics[scale = 0.4]{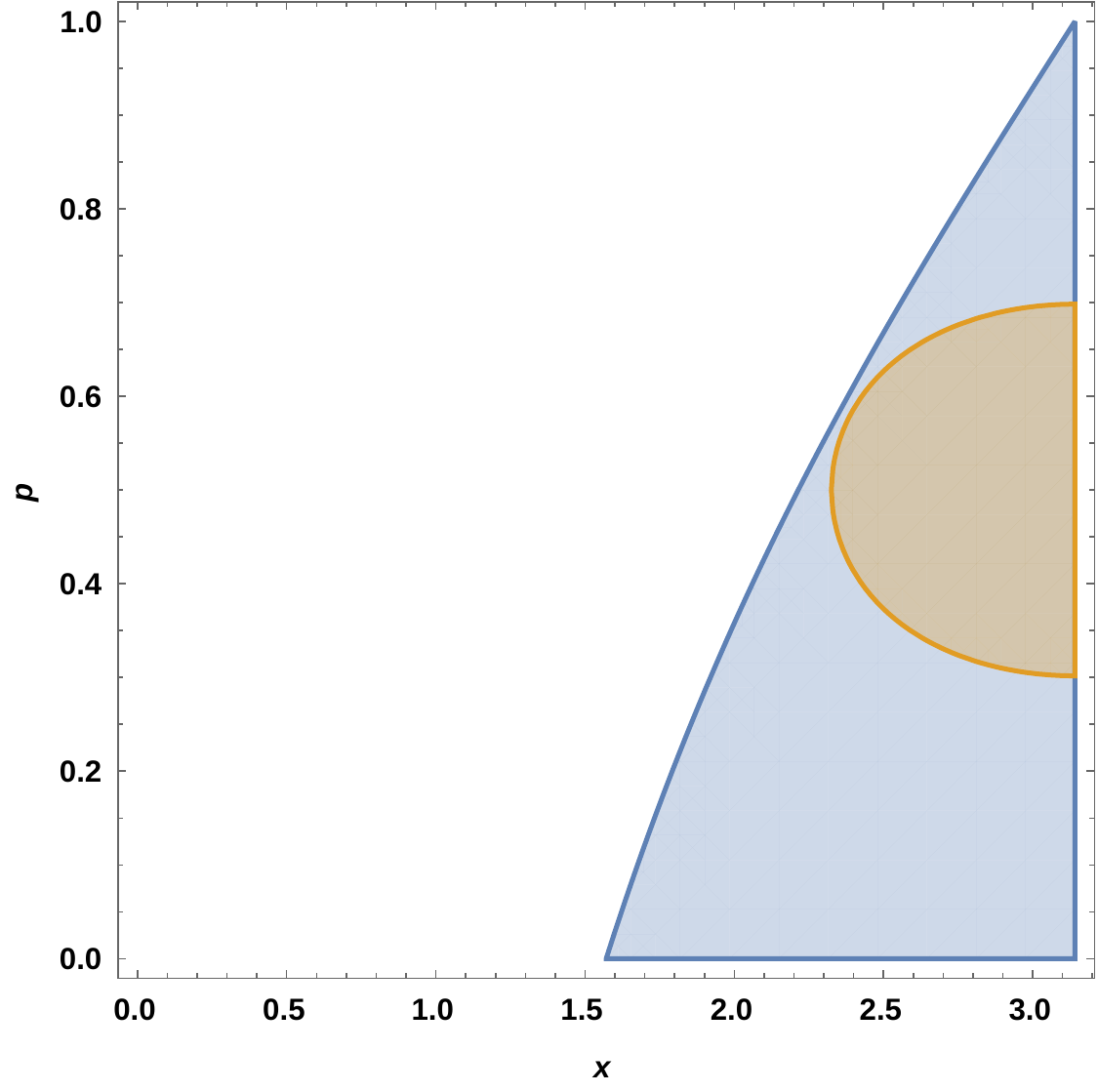}
\caption{(Color online) Regions in $(x,p)$-plane indicates separable and absolutely separable states when $|\psi\rangle$ is sent through local amplitude damping channels. The parameter $x$ corresponds to the state $\ket{\psi}$ and the parameter $p$ corresponds to the channels. The blue (bigger) region is for separability while the yellowish (smaller) region  represents states that are absolutely separable. The parameter $p$  in the vertical axis is a dimensionless quantity, while  the parameter $x$ in the horizontal axis is in radian.}\label{ADC}
\end{figure}

Let us now move to the scenario where two local ADCs are acted on  \(|\psi \rangle\), resulting in an output state of rank-$4$. It turns out that the output state is separable when \(\tan (x/2) \geq 1/(1-p)\). By using condition (\ref{eq1}), we can  find the condition on \(p\) and \(x\) for which the resulting state is absolutely separable. Here also the eigenvalues are independent of the phase of the initial state. 
 
Unlike the depolarizing channels, we find that the final state is absolutely separable only when \(x \geq2.325 \). The criteria for absolute separability are satisfied in the neighborhood of  \(p=0.5\). The range of \(p\) in which the state is absolutely separable increases with the increase of \(x\) and becomes maximum when \(x =\pi\). For example, when  \(x = \pi\), the state is separable in the entire range of  \(p\) although it is absolutely separable when \(p \in [0.302, 0.6998]\).  For the input state  with \(x = 2.4\), the state remains absolutely separable when \(p \in [0.414, 0.586]\) and separable for  \(0 \leq p \leq 0.61  \). Therefore, like depolarizing channels, there also exists a range of parameter \(p\) in which the state is separable, but not absolutely separable as depicted in Fig. \ref{ADC}. However, the absolute separability of the resulting state with respect to amplitude damping noise requires a minimum amount of entanglement in the input state which is in a sharp contrast to DPC. 
 
Interestingly, we find that in case of  amplitude damping channel (ADC),  there exist no random pure states that, after sending through the double-sided local channels become absolutely separable. If we start with rank-2 or higher rank Haar uniformly generated states and ADC acting on both the qubits, such states are produced although unlike depolarizing channels, we observe that the percentage of absolutely separable states grows with the increase in the rank of the initial states.

\subsubsection{Absolutely separable states via phase damping channels}

\begin{figure}
\includegraphics[scale = 0.4]{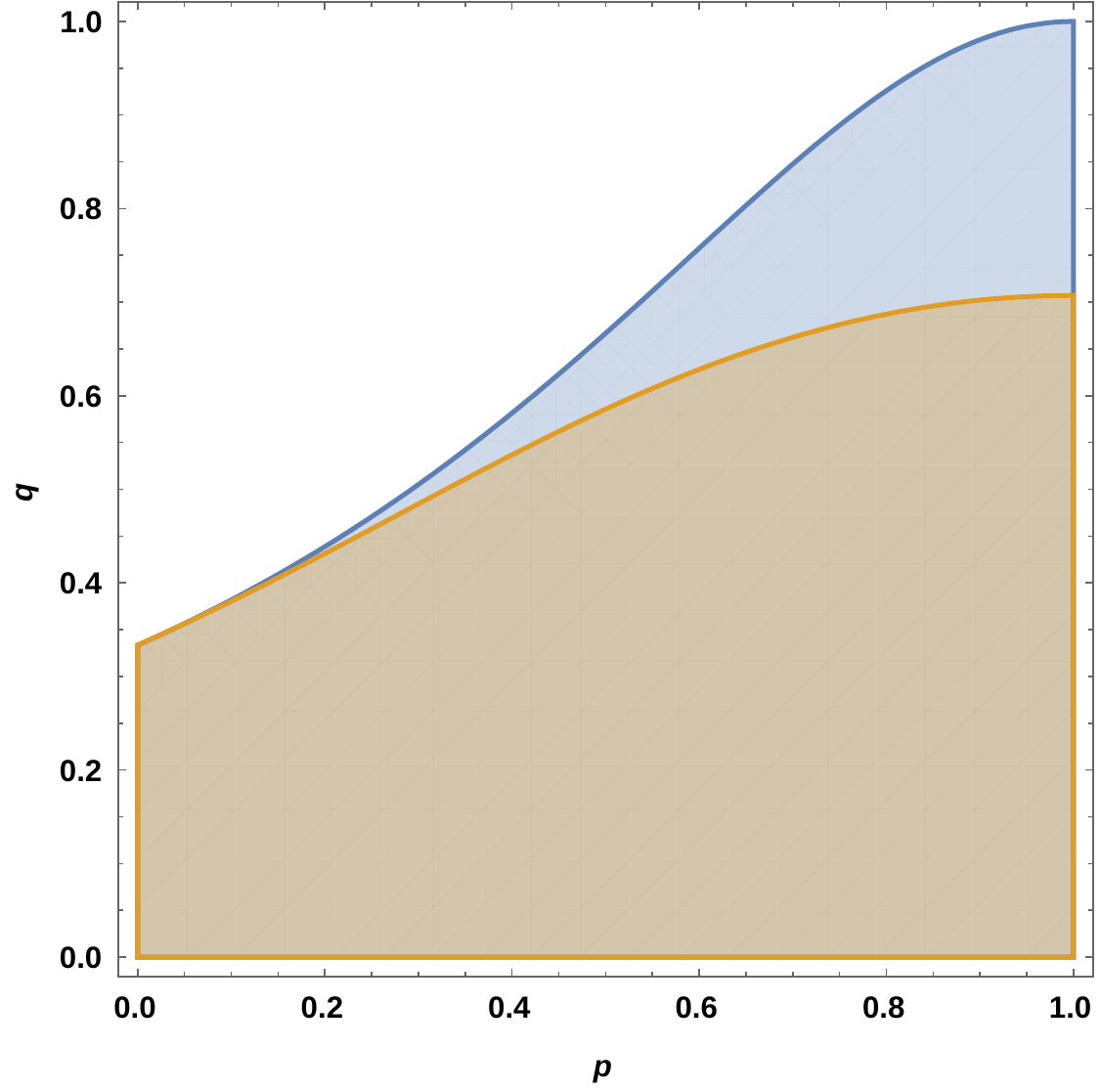}
\caption{(Color online) Map of  separable (blue)  and absolutely separable (yellow) states for local PDC with the Werner state as input. The abscissa and ordinate respectively represent the noise parameter, \(p\), and the mixing parameter,  \(q\), of the Werner state. Both axes are dimensionless.
}
\label{PDC}
\end{figure}

Phase damping channels  make the pure rank-$1$ state to a rank-$2$ one and hence the output state produced after PDCs cannot be absolutely separable, as discussed in Sec. \ref{sec3}. To obtain an absolutely separable state by using PDCs, either we consider a rank-$2$ state which can produce a rank-$3$ state with equal mixing parameter, or we can start with a rank-$3$ or a rank-$4$ state that can be absolutely separable.
 
To illustrate this feature, let us first consider a rank-$3$ state, given by
\begin{equation}\label{eq9}
\rho_3 = \frac{1}{3} |\psi\rangle \langle \psi| + \frac{1}{3} |01\rangle\langle01| + \frac{1}{3}|10\rangle\langle10|,
\end{equation} 
where $\ket{\psi}$ = $\cos{(x/2)}\ket{00}+e^{-i\phi}\sin{(x/2)}\ket{11}$. By Proposition \ref{eq1}, it is an extreme point of {\bf AS}. If one sends the state through local PDCs, it is easy to check that the state becomes rank-$4$ except $x=0$ as well as $x=\pi$ and remains absolutely separable, i.e., it satisfies the condition (\ref{eq1}) in a strict sense. Note that at  $x=0$ and $x=\pi$, the state remains unaffected by the PDCs.
 
As a second example, let us consider the initial state for the PDCs as the Werner state, given by \(q |\phi^+\rangle \langle \phi^+| + (1 - q) \frac{\mathbb{I}}{4}\) which is of rank-$4$.  It is known that the state is separable as well as absolutely separable when \(q \leq 1/3\). When the state is sent through two noisy PDC channels, entanglement gets destroyed and hence the state becomes separable even for \(q> 1/3\). We find that the state becomes separable against $q$ in a bigger range than the value of $q$ below which it is absolutely separable, as shown in Fig.~\ref{PDC}. Therefore, phase damping noise introduces a gap between separable and absolutely separable regions for Werner states. 

Numerical simulations of Haar uniform generation of two-qubit states reveal that no pure, rank-2 and rank-3 states under the action of PDC can produce absolutely separable states, thereby showing its high amount of robustness in the preservation of entanglement. Specifically,  
we find that among $10^4$ rank-4 randomly generated states,  only 60 states can create absolutely separable states when local phase damping channels act on both the qubits. It also indicates that the production of absolutely separable states via depolarizing, amplitude damping and phase damping channels from random input states is capable to distinguish these three channels.
\section{Discussion}
\label{sec5}

Absolutely separable states are those which cannot be made entangled by the action of global unitary operations. Therefore, from the resource theoretic perspective, it is important to study the set of useless states, also known as absolutely separable states. In this work, we considered $2\otimes d$ dimensional states and showed that absolutely separable states of rank-($2d-1$) are all extreme points of the set of such states. We proved that the states with full-rank satisfying strict absolute separability condition are the interior points, otherwise, they are boundary points of the set of absolutely separable states. We also showed that there exist full-rank states which are the boundary  points but not extreme points of the above set. We further proved the existence of full-rank extreme points of the set. 

We showed a possible method to make absolutely separable states entangled  by adding an auxiliary system. We also considered the reverse process, specifically the generation of absolutely separable states with the help of decoherence. In particular, we found the range of  noise parameter which can produce absolutely separable states from entangled states when sent through local noisy channels. We also showed that after sending maximally entangled states via a local depolarizing channel, threshold noise value producing separable and absolutely separable states coincide, while with the decrease of  entanglement content of the input pure state, the gap between these two critical values increases. Moreover, when Haar uniformly generated  two-qubit states are sent through noisy channels, we found that the production of absolutely separable states depends on the rank of the input states, thereby showing a discrimination method for noisy channels.

\section*{ACKNOWLEDGMENTS}
S.M. acknowledges the Ministry of Science and Technology in Taiwan (Grant No.~110-2811-M-006-501). We acknowledge the support from the Interdisciplinary Cyber Physical Systems (ICPS) program of the Department of Science and Technology (DST), India, Grant No.~DST/ICPS/QuST/Theme-1/2019/23.

\bibliography{ref}
\end{document}